\documentclass[journal,12pt,draftclsnofoot,onecolumn]{IEEEtran}
\usepackage{amsmath}
\usepackage{graphicx}
\usepackage{epstopdf}
\usepackage{amssymb}
\usepackage{amsthm}
\usepackage{caption}
\usepackage{cite} 
\usepackage{subfigure}
\usepackage{enumerate}
\usepackage{caption} 
\usepackage{tikz}
\usetikzlibrary{patterns}
\usetikzlibrary{shapes,arrows}
\usepackage{verbatim}
\usetikzlibrary{positioning}
\usetikzlibrary{shapes.geometric}
\usetikzlibrary{shapes.misc}
\usetikzlibrary{patterns}
\begin{document}

\title{Optimal Thresholds for Coverage and Rate in FFR Schemes for Planned Cellular Networks}
\author{Suman Kumar \hspace*{.5in} Sheetal Kalyani \hspace*{.5in} K. Giridhar  \\
\hspace{0in} Dept. of Electrical Engineering \\
      \hspace{0in}IIT Madras \\
  \hspace{-0.7in}     Chennai 600036, India   \\ 
{\tt \{ee10d040,skalyani,giri\}@ee.iitm.ac.in}\\
}
\maketitle
\begin{abstract}
Fractional frequency reuse (FFR) is an  inter-cell interference coordination  scheme  that is being actively researched for emerging wireless cellular networks. In this work, we consider hexagonal tessellation based planned FFR deployments, and derive expressions for the coverage probability and normalized average rate for the downlink. In particular, given reuse $\frac{1}{3}$ (FR$3$ ) and reuse $1$ (FR$1$) regions, and a Signal-to-Interference-plus-noise-Ratio (SINR) threshold $S_{th}$ which decides the user assignment to either the FR$1$ or FR$3$ regions, we theoretically show that: $(i)$ The optimal choice of $S_{th}$ which maximizes the coverage probability is $S_{th} = T$, where $T$ is the required target SINR (for ensuring coverage), and  $(ii)$ The optimal choice of $S_{th}$ which maximizes  the normalized  average rate is given by the expression $S_{th}=\max(T, T')$, where $T'$ is a function of the path loss exponent and the fade parameters. For the optimal choice of $S_{th}$, we show that FFR gives a higher rate than FR$1$ and a better coverage probability than FR$3$.  The impact of frequency correlation over the sub-bands allocated to the FR$1$ and FR$3$ regions is analysed, and it is shown that correlation decreases the average rate of the FFR network. Numerical results are provided, and  these  match with the analytical results.
\end{abstract}
\section{Introduction}
Orthogonal Frequency Division Multiple Access (OFDMA) based system ensures orthogonality among the intra-cell users, but  reuse one OFDMA system deployments suffer from inter-cell interference. Inter-cell interference coordination (ICIC)  schemes  minimize co-channel interference, and simultaneously maximize spatial reuse \cite{4907410} . Fractional frequency reuse (FFR) \cite{5534591} is a simple ICIC scheme, and has been proposed for OFDMA based wireless networks which are based on standards such as IEEE WiMAX \cite{4644118} and 3GPP LTE  \cite{4907406}.

FFR is a combination of frequency reuse $1$ (FR$1$) and frequency reuse $\frac{1}{\delta}$ (FR$\delta$). FR$1$ allocates all the frequencies to each cell, leading to a high spatial reuse, but could result in a low coverage due to inter-cell interference. On the other hand, FR$\delta$ allocates $\frac{1}{\delta}$ of  the frequencies to each cell, and   trades-off spatial reuse and rate, for  higher coverage. FFR exploits the advantage of both FR$1$ and FR$\delta$ by using FR$1$ for the cell-centre users (i.e., for those users who experience less interference from other cells and/or are close to their serving base station (BS)) and FR$\delta$ for the cell-edge users (i.e., for those users who experience high interference from 
 co-channel signals from neighbouring cells and/or are far from their serving BS). Typically, there are two modes of FFR deployment: static and dynamic FFR \cite{4907410}. In this paper, we consider the more practical static FFR scheme, where  all the parameters are configured and kept fixed over a certain period of time.  Fig. \ref{fig:ffr} depicts a frequency allocation in the FFR scheme for three adjacent cells, where $F_1$, $F_2$ and $F_3$ each use $x\%$ of the total spectrum, and $F_0$ uses $(100-3x)\%$ of the spectrum.

FFR schemes have been rather well studied using both system level simulations  and theoretical analysis \cite{6047548, 6399168, 4678075, 6214087, 765583, 4015740, 4525942, 5202320}. Blocking probability of reuse partitioning based cellular system for voice traffic,  assuming  a TDMA system has been derived in \cite{765583}. Analysis of the theoretical capacity and outage rate of an OFDMA cellular system assuming FFR and  proportional fair scheduling  has been presented in \cite{4525942} where the users are classified as cell-centre users and cell-edge users based on the geographical location. The optimization of design parameter (distance threshold\footnote{Based on pre-determined distance from the BS, users are divided into cell-centre users and cell-edge users and here the design parameter is a distance threshold ($R_{th}$).} or SINR threshold\footnote{Based on pre-determined SINR, users are divided into cell-centre users and cell-edge users here the design parameter is a SINR threshold ($S_{th}$).}) of FFR has been studied using graph theory in \cite{5198612}, and convex optimization in \cite{4657213}.  It has been shown in \cite{4657213} that the optimal frequency reuse is FR$3$ for the cell-edge users. The distance and SINR threshold were obtained based on the maximization of the ratio of the average SINR to the variance of SINR in \cite{5202320}, but small scale fading was not considered in their analysis. The optimal size of centre and edge region based on the maximization of throughput is given in \cite{6214087}, and optimization of the size of an  inner radius based on cell-edge efficiency is given in \cite{6399168}.  Recently, the average cell throughput in FFR system is derived in \cite{6225396} as a function of the distance threshold.  In \cite{5493830}, an optimal sub-band allocation for the generalized FFR in OFDMA based networks with irregularly shaped cells is presented. It is shown in \cite{novlan2010comparison}  that there exists  an optimal radius threshold for which the average rate is maximum. An analytical framework to calculate coverage probability and average rate in FFR scheme has been presented in  \cite{6047548}  using Poisson point process (PPP), where they classify users as cell-centre or cell-edge users based on a SINR threshold.

In our understanding, rather than using a SINR threshold as in \cite{6047548}, most of the existing analytical approaches derive coverage probability and rate based on a distance threshold  for user classification as either a cell-centre or cell-edge user. However, in the presence of short term fading and path loss a  significant number of geographically cell-centre users could be declared as cell-edge users  (and vice-versa), when  SINR threshold is used for user classification.  For example in Fig. \ref{fig:fig5} even  at $500$m, $20\%$ to $30\%$ users are declared to be cell-edge users based on their SINR, while based on a simple distance threshold they would have been declared as cell-centre users.  Further, to the best of our knowledge, no prior work has analytically derived the optimal SINR threshold corresponding to the coverage probability and rate in terms of received SINR\footnote{Prior work have obtained the optimal distance threshold ($R_{th}$) corresponding to rate and it has been stated that the optimal SINR threshold also correspond to the optimal distance threshold \cite{6225396}. However in the presence of short term fading  optimal $S_{th}$ need not  correspond to optimal $R_{th}$ as is apparent from Fig.\ref{fig:fig5}.}. While \cite{6047548} considers both coverage probability and rate, the network in \cite{6047548}  assumes an unplanned FFR network, where the cells using the same set of $\frac{1}{\delta}$ frequency are randomly located. Hence, two cells using the same frequency for the cell-edge users could be co-located in  \cite{6047548}. However, with FFR based deployments, the regions using the same frequency are typically planned to be as far apart as possible. Our work assumes such a planned FFR pattern which is more realistic. We consider hexagonal tessellation, with base-stations at the cell-centre and uniform user distribution, and we do not use or depend on the PPP model as in \cite{6047548}.

This work also carefully examines the correlation over the sub-bands (i.e., the frequency resources allocated to $F_0$, $F_1$, $F_2$ and $F_3$ in Fig.\ref{fig:ffr}) used in the FFR system. All prior work on FFR have assumed the sub-bands experience independent fading, which is mathematically convenient, but practically not realisable. Indeed, when we consider block  modulation such as OFDM, the channel delay spread is assumed to be finite in duration and confined to within the cyclic prefix of the OFDM symbol. Such a time limited (typically less than $20\%$ of the useful OFDM symbol duration) impulse response will introduce correlation over the frequency domain. Unless the sub-bands are diversity spaced (i.e., spaced apart in Hz by more than the reciprocal of the delay spread), correlation will exist. Since the delay spread seen on the downlink is user dependent, it is virtually impossible to ensure that the sub-bands used to define the $F_i$s in Fig.\ref{fig:ffr}, are uncorrelated for each user scheduled on the downlink. Therefore, in our analysis we specifically take into account the correlation across the sub-bands while deriving the optimal SINR threshold. Expressions for coverage probability and normalized  average rate in planned FFR networks are  derived and  the following new results are presented:
 \begin{enumerate}[(a)]
\item The optimal SINR threshold  that  maximizes the  coverage probability for FFR  is derived for a given $T$. We show that the optimal $S_{th}$ (denoted by $S_{opt,C}$) is $S_{th} = T$, and if one choose the SINR threshold  to be $S_{opt,C}$, then   the  coverage probability for FFR  is higher than that for FR$3$. The increase in FFR coverage probability over the FR$3$ coverage probability is due to sub-band diversity  gain which is accrued by the system when a user get classified as either a cell-centre or a cell-edge user.
\item The optimal SINR threshold  that  maximizes the normalized average rate for FFR is derived for a given $T$. We show that the optimal $S_{th}$ (denoted by $S_{opt,R}$) is equal to $\max(T,T')$, where $T'$ is a fixed SINR value, which depends on the system parameters such as path loss factor, fading parameters, etc. Corresponding to the $S_{opt,R}$, the normalized average rate for FFR is higher than the rates provided by either  FR$1$ or FR$3$. 
 \end{enumerate}
Numerical results are also provided along with our analytical results, and both are seen to be in close agreement.
  \section{System Model}
A homogeneous macrocell network with  hexagonal tessellation  with inter cell site distance $2R$ is considered as shown in Fig. \ref{fig:hexagonal}. The SINR $\eta(r)$ of a user located at $r$ meters from its serving BS is given by
\begin{equation} 
\eta(r)=\frac{gr^{-\alpha}}{\frac{\sigma^2}{P} + I},\text{ } \text{ }\text{ }\text{ }I=\underset{i \in \psi}\sum   h_i d_i^{-\alpha}. \label{eq:int}
\end{equation}
where the transmit power of a BS is denoted by $P$ and $\psi$ denotes the set of all interfering BSs. A standard path loss model $\|x\|^{-\alpha}$ is assumed, where $\alpha\geq 2$ is path loss exponent, and $\|x\|$ is the distance of a user from the BS similar to \cite{6047548}. Note that it is assumed that users are at least a distance $d$ away from the BS\footnote{Typically, the path loss model is assumed to be $\max\{d,\|x\|\}^{-\alpha}$.}. Noise power is denoted by $\sigma^2$. Here, $r$ and $d_i$ are the distances from the user to the serving BS and $i^{th}$ interfering  BS, respectively, and $g$ and $\{h_i\}$ denote the corresponding channel fading power. Here, $\{h_i\}$ denote the set of $h_i$s $\forall$ $i\in\psi$. These $g$ and $\{h_i\}$  are independent and identically exponentially distributed (i.i.d.) with unit mean, i.e, $g\sim \exp(1)$ and $h_i\sim \exp(1)\forall i$. 
\begin{figure}[ht]
\centering
\begin{tikzpicture}
\node[pattern=north west lines, pattern color=red!40, regular polygon, regular polygon sides=6,minimum width= 2  cm, draw] at (5,5) {};
\node[pattern=north west lines, pattern color=red!40, regular polygon, regular polygon sides=6,minimum width=2 cm, draw] at (5+1.5*1,5-0.866*1) {};
\node[pattern=north west lines, pattern color=red!40, regular polygon, regular polygon sides=6,minimum width=2 cm, draw] at (5+1.5*1,5+0.866*1) {};
\draw [ fill=white!10,line width=.01cm](5,5) circle (0.45);
\draw [fill=white!10,line width=.01cm](5+1.5*1,5-0.866*1) circle (0.45);
\draw [fill=white!10,line width=.01cm](5+1.5*1,5+0.866*1) circle (0.45);
\node at (5,5){$F_0$};
\node at (5,5.65){$F_1$};
\node at (5+1.5*1,5-0.866*1){$F_0$};
\node at (5+1.5*1,5-0.866*1+0.65){$F_2$};
\node at (5+1.5*1,5+0.866*1){$F_0$};
\node at (5+1.5*1,5+0.866*1+0.65){$F_3$};
\draw [black,line width=.02cm](8,5) rectangle (9,6);
\draw [black,line width=.02cm](8,4) rectangle (9,5);
\draw[pattern=north west lines, pattern color=red!40](8,4) rectangle (9,5);
\node at (10,5.5){Cell-centre};
\node at (10,4.5){Cell-edge};
\node at (7,7.5){FFR};
\end{tikzpicture}
\caption{Frequency allocation in FFR for three neighbouring cells with $\delta=3$. The cell-centre users of all the cells use a common frequency band $F_0$, while the cell-edge users of the three cells use different frequency bands, namely $F_1$, $F_2$ and $F_3$.}
  \label{fig:ffr}
        \end{figure}
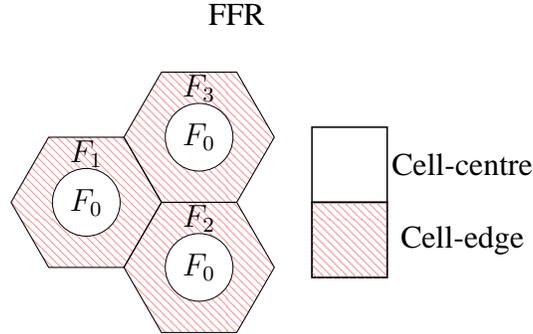
 \begin{figure}[ht]
 \centering
 \includegraphics[scale=0.3]{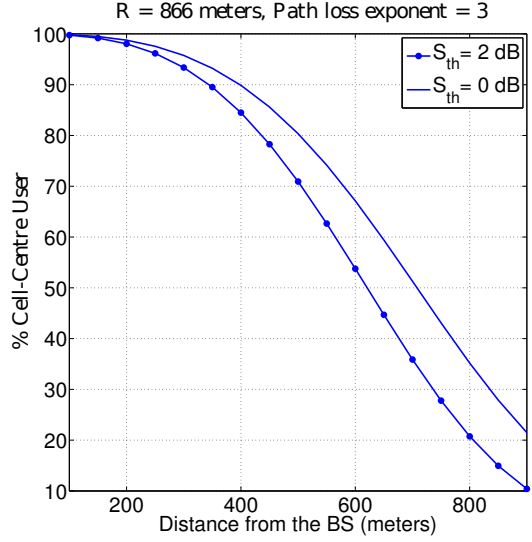}
 \caption{Percentage of  cell-centre users based on the SINR threshold ($S_{th}$) with respect to distance from the BS.}
 \label{fig:fig5}
 \end{figure}
\begin{figure}[ht]
 \centering
 \includegraphics[scale=0.3]{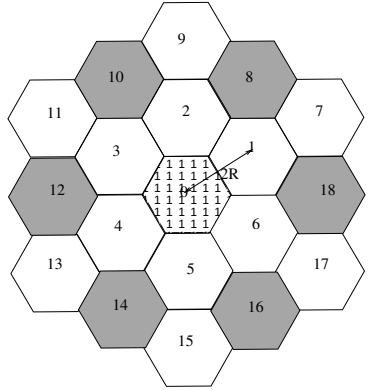}
 \caption{Hexagonal structure of 2-tier macrocell. Interference for $0$th cell in FR$1$ system is contributed form  all the neighbouring $18$ cells, while in a  FR$3$ system it is contributed only from the shaded cells.}
 \label{fig:hexagonal}
 \end{figure} 

 Similar to  \cite{6047548}, the users are classified  as cell-centre users and cell-edge users based on the received SINR at the mobile station. If the calculated SINR for a user is less than the specified  SINR threshold $S_{th}$,  the user is classified as a cell-edge user.  Otherwise,  the user is  classified as a cell-centre user. Typically, FFR  divides the whole frequency into a total of $1+\delta$ parts where one part is common to all the cells for the cell-centre users. One among the $\{1,\cdots,\delta \} $ parts is assigned to cell-edge users in each cell in a planned fashion. Due to physical movement, it is possible that a cell-centre user can get re-assigned as a cell-edge user (or vice versa). In such a case, the channel gain corresponding to the desired signal and the interfering signal change to $\hat{g}$ and $\{\hat{h}_i\}$. Now based on the coherence bandwidth of the OFDM system, and the bands associated with $F_0$ to $F_3$ as shown in Fig.\ref{fig:ffr}, it is possible that $\hat{g}$ and $\{\hat{h}_i\}$ are either correlated or uncorrelated with $g$ and $\{h_i\}$, respectively. This would depend on the particular user's channel conditions and the instantaneous coherence bandwidth with respect to the FFR frequency bands. To better understand the impact of correlation among the sub-bands on FFR performance, in this paper, we consider following two extreme cases:

{\em Case $1$:} $g$ and $\hat{g}$ are uncorrelated and also  $h_i$ and $\hat{h}_i$, are uncorrelated $\forall i$.

{\em Case $2$:} $g$ and $\hat{g}$ are fully correlated and also  $h_i$ and $\hat{h}_i$, are fully correlated $\forall i$.

In reality these channel power may be partially correlated, but the analysis of partial (arbitrary) correlation  is quite complicated and beyond the scope of this work. However, the analysis of the above two extreme cases we believe, is sufficient to understand the impact of correlation among sub-bands.

\section{Coverage Probability Analysis of FFR }
Coverage probability  is defined as the probability that a randomly chosen user's instantaneous SINR $\eta(r)$ is greater than  $T$ and is given by
\begin{equation}
 CP(T,r)= P[\eta(r)>T].
\end{equation}
The coverage probability of a user who is at a distance $r$ meters from the BS in a FR$1$ network is given by
\begin{equation}
CP_1(T,r)=P[\eta(r)>T]= P\left[g>T r^{\alpha}I+T r^{\alpha}\frac{\sigma^2}{P}\right] \label{eq:fr1},
\end{equation}
where $I$ is defined in \eqref{eq:int}. Since $g\sim \exp(1)$, $h_i\sim \exp(1)$, and $h_i$ are i.i.d., $CP_1(T,r)$ is given by
\begin{equation} 
CP_1(T,r)=E_{h_i}\left[e^{-T r^\alpha I-T r^{\alpha}\frac{\sigma^2}{P}}\right] = \underset{i \in \psi}\prod E_{h_i}\Big[e^{-T  r^\alpha  h_i d_i^{-\alpha} }\Big]e^{-T r^{\alpha}\frac{\sigma^2}{P}}=\underset{i \in \psi}\prod \frac{1}{1+T  r^\alpha  d_i^{-\alpha} }e^{-T r^{\alpha}\frac{\sigma^2}{P}} \label{eq:reuse1},
\end{equation}
where $\psi$ is the set of interfering BS in a FR$1$ network. Similarly, the coverage probability of a user located at distance $r$ meters from the BS in a FR$3$ network can be given by
\begin{equation}
CP_3(T,r) =\underset{i \in \phi}\prod \frac{1}{1+T  r^\alpha  d_i^{-\alpha} }e^{-T r^{\alpha}\frac{\sigma^2}{P}}
\label{eq:reuse3}
 \end{equation}
where $\phi$ the set of interfering cells\footnote{Note that the effect of interferers beyond tier $2$ is accounted for in the noise variance ($\sigma^2$) term.} for FR$3$ scheme is a function of the reuse plan. In the two tier planned macrocell network shown in the Fig. \ref{fig:hexagonal}, $\phi=\{8, 10, 12, 14, 16, 18\}$. Now, we derive the coverage probability for FFR for both the extreme cases of correlation mentioned earlier.
\subsection{ Case $1 :$  $g$ and $\hat{g}$ are uncorrelated  and $h_i$ and $\hat{h}_i$ are also uncorrelated $\forall i$}
The coverage probability of cell-centre user who is at distance $r$ meters from the $0^{th}$ BS in a FFR network $ CP_{F,c}(r)$ is given by  
 \begin{equation*}
 CP_{F,c}(r)\stackrel{(a)}=P[\eta(r)>T|\eta(r)>S_{th}]=P\left[\frac{g r^{-\alpha}}{I+\frac{\sigma^2}{P}}>T\Big{|}\frac{g r^{-\alpha}}{I+\frac{\sigma^2}{P}}>S_{th}\right],
 \end{equation*}
 where, $(a)$ follows from the fact that for a cell-centre user  SINR $\geq S_{th}$. Applying Bayes' rule, one can rewrite $CP_{F,c}(r)$ as
 \begin{equation}
CP_{F,c}(r)= \frac{P\Big[\frac{g r^{-\alpha}}{I+\frac{\sigma^2}{P}}>T,\frac{g r^{-\alpha}}{I+\frac{\sigma^2}{P}}>S_{th}\Big]}{P\Big[\frac{g r^{-\alpha}}{I+\frac{\sigma^2}{P}}>S_{th}\Big]}=\frac{\underset{i \in \psi}\prod\frac{1}{1+\max\{T, S_{th}\}r^\alpha d_i^{-\alpha}}e^{-\max\{T, S_{th}\} r^{\alpha}\frac{\sigma^2}{P}}}{\underset{j \in \psi}\prod\frac{1}{1+S_{th}r^\alpha d_j^{-\alpha}}e^{-S_{th} r^{\alpha}\frac{\sigma^2}{P}}}.\label{eq:centre}
\end{equation}
Similarly, the coverage probability of a  cell-edge user who is at a distance of $r$ meters from the BS in FFR network $CP_{F,e}(r)$ is given by
\begin{equation*}
CP_{F,e}(r)=P[\hat{\eta}(r)>T|\eta(r)<S_{th}]=\frac{P\Big[\frac{\hat{g} r^{-\alpha}}{\hat{I}+\frac{\sigma^2}{P}}>T,\frac{g r^{-\alpha}}{I+\frac{\sigma^2}{P}}<S_{th}\Big]}{P\Big[\frac{g r^{-\alpha}}{I+\frac{\sigma^2}{P}}<S_{th}\Big]} \label{eq:ffr_edge}.
\end{equation*}
Here, the cell-edge user will experience a new interference $\hat{I}=\underset{i \in \psi}\sum   \hat{h}_i d_i^{-\alpha}$ and new channel power $\hat{g}$, i.e, a new SINR $\hat{\eta}(r)$ due to  the fact that  the cell-edge user is now a  FR$3$ user. Basically, $\hat{\eta}(r)$ denotes the SINR experienced by the user at a distance of $r$ meters from the BS in a FR$3$ system. Since $g$ and $\hat{g}$ are i.i.d and  $h_i$ and $\hat{h}_i$ are also assumed to be i.i.d and hence $CP_{F,e}(r)$ can be simplified as

\begin{equation}
CP_{F,e}(r)=P\left[\frac{\hat{g} r^{-\alpha}}{\hat{I}+\frac{\sigma^2}{P}}>T\right]=CP_3(T,r) \label{eq:ffr_edge1}.
\end{equation}
\newtheorem{theorem}{Lemma}
We now derive the coverage probability of an user in the FFR network. This $CP_f(r)$ can be written as
\begin{equation}
CP_f(r)=CP_{f,c}(r)P[\eta(r)>S_{th}]+CP_{f,e}(r)P[\eta(r)<S_{th}]\label{eq:ffr}
\end{equation} 
Here, the first term denotes the coverage probability contributed by cell-centre users and the second term denotes the contribution from the cell-edge users. The above expression can be simplified as
\begin{equation}
\textstyle
CP_{F}(r)= \underset{i \in \psi}\prod\frac{1}{1+\max\{T, S_{th}\}r^\alpha d_i^{-\alpha}}e^{-\max\{T, S_{th}\} r^{\alpha}\frac{\sigma^2}{P}}+CP_3(T,r)-CP_3(T,r)CP_1(S_{th},r)\label{eq:edge5},
\end{equation}
by using the expression in \eqref{eq:centre} for $CP_{f,c}(r)$ and the expression in \eqref{eq:ffr_edge1} for $CP_{f,e}(r)$.
\begin{theorem}
The optimum $S_{th}$ (denoted by $S_{opt,C}$) that maximizes the   FFR coverage probability is $S_{th}= T$, and when  SINR threshold is taken  to be $S_{opt,c}$, the coverage probability of FFR achieves  higher coverage  than that of FR$3$.
\end{theorem}
\begin{proof}
To obtain the $S_{opt,C}$, we consider the following three possibilities: $(i)$ $S_{th}<T$, $(ii)$ $S_{th}=T$, $(iii)$ $S_{th}>T$.\\
$(i)$ $S_{th}<T$:$ \text{ Let }S_{th}=T-\Delta, $ where $\Delta>0$, then $CP_f(r)$ in terms of $T$ can be given by
\begin{equation}
\textstyle
CP_{F}(r, S_{th}< T)= \underset{i \in \psi}\prod\frac{1}{1+Tr^\alpha d_i^{-\alpha}}e^{-T r^{\alpha}\frac{\sigma^2}{P}}+CP_3(T,r)-CP_3(T,r)CP_1(T-\Delta,r)\label{eq:edge8}.
\end{equation}
$(ii)$ $S_{th}=T$: In this case $CP_f(r)$ in terms of $T$  can be given by
\begin{eqnarray}
CP_{F}(r, S_{th}= T)&&= \underset{i \in \psi}\prod\frac{1}{1+Tr^\alpha d_i^{-\alpha}}e^{-T r^{\alpha}\frac{\sigma^2}{P}}+CP_3(T,r)-CP_3(T,r)CP_1(T,r)\label{eq:edge9}.\\
&&= CP_1(T,r)(1-CP_3(T,r)) +CP_3(T,r)\label{eq:edge91}.
\end{eqnarray}
$(iii)$ $S_{th}>T$:$ \text{ Let }S_{th}=T+\Delta, $ where $\Delta>0$, then $CP_f(r)$ in terms of $T$ can be given by
\begin{eqnarray}
CP_{F}(r, S_{th}> T)  &= \underset{i \in \psi}\prod\frac{1}{1+(T+\Delta)r^\alpha d_i^{-\alpha}}e^{-(T+\Delta) r^{\alpha}\frac{\sigma^2}{P}}+CP_3(T,r)-CP_3(T,r)CP_1(T+\Delta,r)\nonumber.\\ 
&= CP_1(T+\Delta,r)(1-CP_3(T,r)) +CP_3(T,r)\label{eq:edge10}.
\end{eqnarray}
Now, we compare the FFR coverage probability for $S_{th}< T$ and $S_{th}=T$ given by \eqref{eq:edge8} and \eqref{eq:edge9}, respectively. Since $CP_1(T-\Delta,r)>CP_1(T,r)$, this implies $CP_{F}(r, S_{th}< T)<CP_{F}(r, S_{th}= T)$. Similarly, we compare the FFR coverage probability for $S_{th}= T$ and $S_{th}>T$ given by \eqref{eq:edge91} and \eqref{eq:edge10}, respectively. Since   $CP_1(T+\Delta,r)<CP_1(T,r)$, this implies $CP_{F}(r, S_{th}= T)>CP_{F}(r, S_{th}> T)$.  Thus, FFR achieves the  maximum coverage  when $S_{th}= T$. Note that when one choose SINR threshold to be $S_{opt,C}$ then the coverage probability of FFR  is higher than FR$3$ coverage probability since $CP_{F}(r, S_{th}= T)= CP_1(T,r)(1-CP_3(T,r)) +CP_3(T,r)>CP_3(T,r)$. The reason for such a behaviour is as follows: only  users having low SINR (low fading gain for the desired signal and/or high fading gain for the interfering signal) move to the cell-edge region and they experience a new independent fading gain at the cell-edge region.  In other words, the increase in FFR coverage probability over the FR$3$ coverage probability is due to sub-band diversity  gain which is accrued by the system when users move from cell-centre to cell-edge.
\end{proof}
We  also numerically evaluate the coverage probability for $S_{th}> T$, $S_{th}= T$,  and $S_{th}<T$ and show that the numerical values match with our theoretical observation. Fig. \ref{fig:fig} shows the coverage probability of cell-centre user, cell-edge user, and a user in FFR for a fixed $T$ and three different values of $S_{th}$ where the coverage probability of a cell-edge user for  all $3$ values of $S_{th}$ is  equal to FR$3$ coverage. It is also observed that the coverage probability for FFR is maximum when $S_{th} = T$, and it is higher than the FR$3$ coverage probability as shown in Fig. \ref{fig:fig}.
  \begin{figure}[ht]
   \centering
   \includegraphics[scale=0.35]{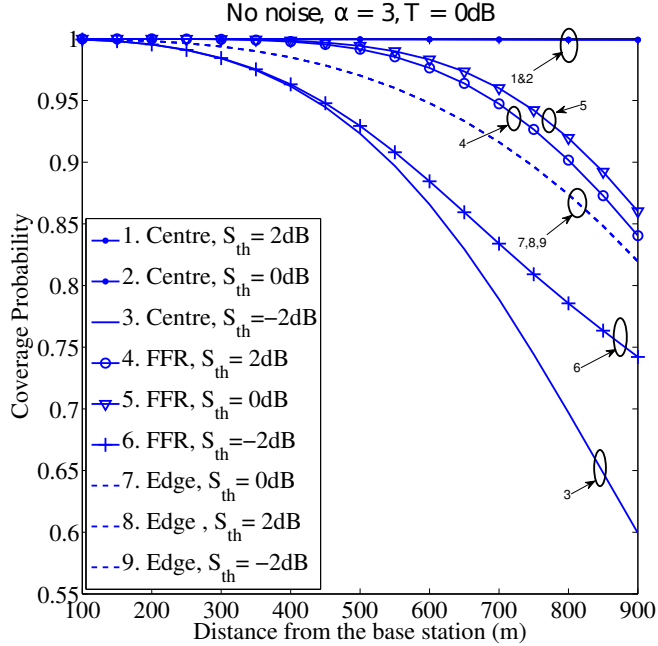}
  \caption{Coverage probability of cell-centre, cell-edge, and FFR  user with respect to distance from BS for different value of SIR Threshold $S_{th}$ when fading are independent across the sub-bands.}
  \label{fig:fig}
  \end{figure}
\subsection{Case $2 :$ $g$ and $\hat{g}$ are completely correlated  and $h_i$ and $\hat{h}_i$ are also completely correlated $\forall i$}
Note that the centre coverage probability is same for both the cases (case $1$ and case $2$) since a user does not change the sub-band when it is a cell-centre user. However, the edge coverage probability is different in case $1$ and case $2$, and in this case the coverage probability of a cell-edge user who is at distance $r$ meters from the BS in FFR network $CP_{F,e}(r)$ is given by
\begin{equation}
CP_{F,e}(r)=P[\hat{\eta}(r)>T|\eta(r)<S_{th}]=\frac{P[\hat{\eta}(r)>T,\eta(r)<S_{th}]}{P[\eta(r)<S_{th}]} \label{eq:ffr_edge2}.
\end{equation}
Putting the value of $CP_{f,c}$ and $CP_{f,e}$ from \eqref{eq:centre} and \eqref{eq:ffr_edge2} into Eq. \eqref{eq:ffr}, the coverage probability in FFR network $CP_f(r)$ can be written as
\begin{equation}
CP_f(r)=\underset{i \in \psi}\prod\frac{1}{1+\max\{T, S_{th}\}r^\alpha d_i^{-\alpha}}e^{-\max\{T, S_{th}\} r^{\alpha}\frac{\sigma^2}{P}}+P[\hat{\eta}(r)>T,\eta(r)<S_{th}]
\end{equation} 
Since $g$ and $\hat{g}$, are completely correlated and $h_i$ and $\hat{h}_i$ are also completely correlated $\forall i$ and hence we use following transformation to further simplify $CP_f(r)$:
\begin{equation}
P[\hat{\eta}(r)>T,\eta(r)<S_{th}]=P[\hat{\eta}(r)>T,\hat{\eta}(r)<\hat{S}_{th}].\label{eq:tran}
\end{equation}
Basically instead of marking a user as a cell-edge user based on the $\eta(r)$ (FR$1$ SINR), we mark on the basis of $\hat{\eta}(r)$ (FR$3$ SINR) by introducing a new SINR threshold $\hat{S}_{th}$. The threshold $\hat{S}_{th}$ is computed using the relation $P[\eta(r)<S_{th}]=P[\hat{\eta}(r)<\hat{S}_{th}]$. This makes sure that same user is marked as cell-edge user with both reuse patterns (FR$1$ and FR$3$).
Now, using the transformation given in \eqref{eq:tran}, $CP_f(r)$ can be  simplified as
\begin{equation}
CP_f(r)=\underset{i \in \psi}\prod\frac{1}{1+\max\{T, S_{th}\}r^\alpha d_i^{-\alpha}}e^{-\max\{T, S_{th}\} r^{\alpha}\frac{\sigma^2}{P}}+P[\hat{\eta}(r)>T]-P[\hat{\eta}(r)>\max\{\hat{S}_{th},T\}]
\end{equation}
In this case, to obtain the $S_{opt,C}$, we consider the following two possibilities: $(i)$ $S_{th}\geq T$, $(ii)$ $S_{th}<T$.\\
$(i)$ $S_{th}\geq T$: In this case,  $CP_f(r)$ in terms of $T$ can be given by
\begin{equation}
\textstyle
CP_{F}(r, S_{th}\geq T)= \underset{i \in \psi}\prod\frac{1}{1+S_{th}r^\alpha d_i^{-\alpha}}e^{-S_{th} r^{\alpha}\frac{\sigma^2}{P}}+CP_3(T,r)-CP_3(\hat{S}_{th},r)\label{eq:edge11}.
\end{equation}
Since $CP_3(\hat{S}_{th},r)=CP_1(S_{th},r)$ and $CP_1(S_{th},r)=\underset{i \in \psi}\prod\frac{1}{1+S_{th}r^\alpha d_i^{-\alpha}}e^{-S_{th} r^{\alpha}\frac{\sigma^2}{P}}$, hence 
\begin{equation}
CP_{F}(r, S_{th}\geq T)=CP_3(T,r).\label{eq:edge1}
\end{equation}
$(ii)$ $S_{th}<T$: In this case $CP_f(r)$ in terms of $T$  can be given by 
\begin{equation}
\textstyle
CP_{F}(r, S_{th} < T)= \underset{i \in \psi}\prod\frac{1}{1+Tr^\alpha d_i^{-\alpha}}e^{-T r^{\alpha}\frac{\sigma^2}{P}}+CP_3(T,r)-CP_3(\max\{\hat{S}_{th},T\},r)\label{eq:edge12}. 
\end{equation}
Note that when $S_{th}<T$, $\hat{S}_{th}$ could be higher or lower than $T$.  When $\hat{S}_{th}>T$,
\begin{equation}
CP_3(\max\{\hat{S}_{th},T\},r)=CP_3(\hat{S}_{th},r)=CP_1(S_{th},r)>CP_1(T,r)\text{ since } S_{th} < T.
\end{equation}
And when $\hat{S}_{th}<T$,
\begin{equation}
CP_3(\max\{\hat{S}_{th},T\},r)=CP_3(T,r)>CP_1(T,r).
\end{equation}
And hence,  
\begin{equation}
\textstyle
CP_{F}(r, S_{th} < T)= \underset{i \in \psi}\prod\frac{1}{1+Tr^\alpha d_i^{-\alpha}}e^{-T r^{\alpha}\frac{\sigma^2}{P}}+CP_3(T,r)-CP_3(\max\{\hat{S}_{th},T\},r)<CP_3(T,r)\label{eq:edge13}. 
\end{equation}
Comparing the FFR coverage probability for $S_{th}\geq T$ and $S_{th}<T$   given by \eqref{eq:edge1} and \eqref{eq:edge13}, respectively, it is apparent that $CP_{F}(r, S_{th} \geq T)> CP_{F}(r, S_{th} < T)$. In other words, when fading are fully correlated across the sub-bands the optimal choice of SINR threshold is $S_{th}\geq T$ and at optimal SINR threshold FFR achieves FR$3$ coverage probability. Unlike the case $1$, FFR coverage probability achieves no gain over FR$3$ coverage probability since there is no sub-band diversity gain while a user moves from cell-centre to cell-edge region.

We now present simulation result of  FFR coverage probability for  different channels (pedestrian A and vehicular A channel \cite{molisch2011wireless}) corresponding to optimal  SINR threshold as shown in Fig. \ref{fig:fig6}. For simulation, a system based on the LTE standard is considered where, corresponding to $5$MHz bandwidth, sampling rate $=7.68$ MHz is considered. The power delay profile of the channels are given in Table $1$ and Table $2$ are as defined in \cite{3gpp1}.
\begin{table}[ht]
        \centering
\caption{Pedestrian A Channel} 
\renewcommand{\tabcolsep}{0.35cm}
\renewcommand{\arraystretch}{1.5}
\begin{tabular}{|c| c| c| c|c | }
\hline 
Relative delay (ns) &0 &110 & 190 &410 \tabularnewline
\hline 
Relative power(dB) & 0.0 & -9.7&-19.2 &-22.8 \tabularnewline
\hline 
\end{tabular}
\end{table} 

\begin{table}[ht]
        \centering
\caption{Vehicular-A Channel} 
\renewcommand{\tabcolsep}{0.35cm}
\renewcommand{\arraystretch}{1.5}
\begin{tabular}{|c| c| c| c|c | c|c | }
\hline 
Relative delay (ns) &0 &310 & 710 &1090 &1730 & 2510\tabularnewline
\hline 
Relative power(dB) & 0.0 & -1.0&-9.0 &-10.0 &-15.0 &-20.0 \tabularnewline
\hline 
\end{tabular}
\end{table} 
It can be observed that coverage probability of both the channels, i.e., pedestrian A and vehicular A channel  is upper bounded by the case $1$ when $g$ and $\hat{g}$ are uncorrelated  and $h_i$ and $\hat{h}_i$ are also uncorrelated $\forall i$  and lower bounded by the case $2$ when $g$ and $\hat{g}$ are completely correlated  and $h_i$ and $\hat{h}_i$ are also completely correlated $\forall i$.
 \begin{figure}[ht]
 \centering
 \includegraphics[scale=0.35]{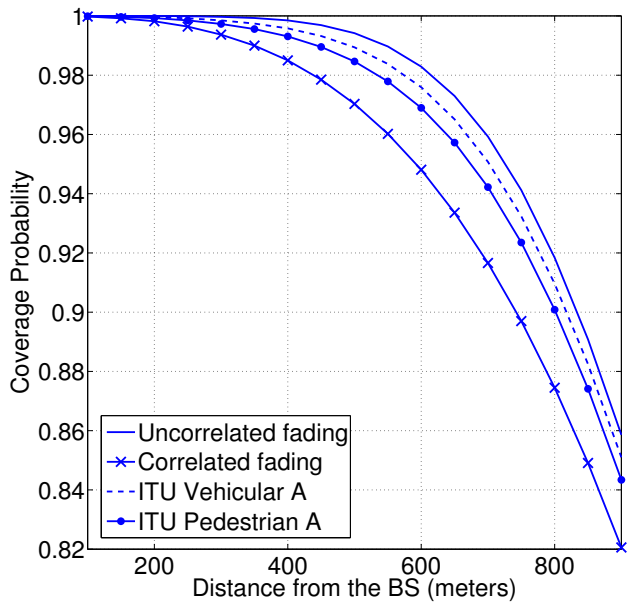}
 \caption{FFR coverage probability of different set of channels at optimal SINR threshold. Here $\alpha=3$, $T=0$dB are assumed.}
 \label{fig:fig6}
 \end{figure}
\section{Normalized average Rate}
In this section, we derive the normalized average rate of FFR scheme and find the optimum value of $S_{th}$ (denoted by $S_{opt,R}$) for which the normalized average rate is maximum. The average  rate of the system  $E[\ln(1+\text{SINR})]$ is not sensitive to the fact  whether users are in coverage or outage. Therefore to take into account the users coverage probability, the average rate is computed in \cite{6171996}, as $E[\ln(1+\text{SINR})|\text{SINR}>T]$. However, this metric does not give us an indication of impact of those users who are not in coverage on the average rate in a system. The rate metric\footnote{ Note that if average rate is defined as $E[\ln(1+\text{SINR})|\text{SINR}>T]$ then with increasing $T$ this average rate only reflects the rate of a few users in the system.} should reflect the fact that with increasing $T$,  the number of users not in coverage also increases.  Hence, we define the normalized average rate as $E[\ln(1+\text{SINR})|\text{SINR}>T]P[\text{SINR}>T]$. 
This metric assign a rate of zero to the users not in coverage and the average rate ($E[\ln(1+\text{SINR})]$) for users in coverage. Hence with increasing $T$ this rate reduces and it actually reflects the average rate behaviour in the system.

In order to derive the normalized average rate\footnote{An Interference limited system is assumed for simplicity. However the derivation of normalized average rate can be easily extended to the case where thermal noise is also considered.} for a FR$1$ and FFR system, we need to calculate the probability density function (pdf) of $r$ which is the distance between $0^{th}$ BS (serving BS) and the desired user. To calculate this pdf, we model the cell shape as an in radius circle, which is considered a fairly good approximation for hexagonal shape \cite{lte}, and assume that the users are uniformly distributed. Therefore, the pdf of $r$ $f_R(r)$ is given by
\begin{equation}
f_R(r)= \left\{
\begin{array}{rl}
&\frac{2r}{R^2}, r\leqslant R\\
&0,  r>R.
\end{array} \right.
\end{equation}
We now derive the normalized average rate for the planned FFR network.
\subsection{Normalized Average Rate in FR$1$ and FR$3$ Systems}
The average rate of a  user at a distance $r$ is $E[\ln(1+\eta(r))]$.
However, the normalized average rate at a distance $r$  in a FR$1$ system is  $ R_1(r)=E[\ln(1+\eta(r))|\eta(r)>T]P[\eta(r)>T]$.  Using the  fact that for a positive random variable $X=\ln(1+\eta(r))$, $E[X]=\int_{t>0}P(X>t)\text{d}t$, $R_1(r)$ can be rewritten as 
\begin{eqnarray}
R_1(r)=\underset{t>0}\int P[\ln(1+\eta(r))>t|\eta(r)>T]P[\eta(r)>T]\text{d}t \label{rate1}.
\end{eqnarray}
Since $\ln(1+\eta(r)) $ is a monotonic increasing function of $\eta(r)$, hence,
\begin{eqnarray}
P[\ln(1+\eta(r))>t|\eta(r)>T]=P[\eta(r)>e^t-1|\eta(r)>T]\stackrel{(b)}=\frac{P[\eta(r)>e^t-1,\eta(r)>T]}{P[\eta(r)>T]} \label{rate2},
\end{eqnarray}
here $(b)$ follows from Bayes' rule. Using \eqref{rate2} and simplifying \eqref{rate1}, $R_1(r)$ can be written as
\begin{eqnarray}
R_1(r)&=&\underset{t>0}\int P[\eta(r))>\max\{e^t-1,T\}]\text{d}t=\underset{t>0}\int P[g>\max\{e^t-1,T\} r^{\alpha}I]\text{d}t \label{rate_1}.
\end{eqnarray}
Using \eqref{eq:fr1} and \eqref{eq:reuse1}, Eq. \eqref{rate_1} can be further simplified as
\begin{equation}
R_1(r)=\underset{t>0}\int\underset{j \in \psi}\prod\frac{1}{1+\max\{e^t-1,T\} r^\alpha d_j^{-\alpha}}\text{d}t.
\end{equation}
 Now, to obtain the normalized average rate in FR$1$ system, spatial average can taken and $R_1$ can be expressed as
\begin{equation}
R_1=\int\limits_{0}^{R}\underset{t>0}\int\underset{j \in \psi}\prod\frac{1}{1+\max\{e^t-1,T\} r^\alpha d_j^{-\alpha}}\text{d}tf_R(r) \text{d}r \label{rate3}.
\end{equation}
The normalized average rate in  FR$3$  can be obtained in a similar fashion and is given by
\begin{equation}
R_3=\int\limits_{0}^{R} \underset{t>0}\int\underset{i \in \phi}\prod\frac{1}{1+\text{max}\{e^t-1,T\}r^\alpha d_i^{-\alpha}}\mathrm{d}tf_R(r)\mathrm{d}r.\label{fr3}
\end{equation}
\subsection{Normalized Average Rate of FFR System When $g$ and $\hat{g}$ are Uncorrelated  and $h_i$ and $\hat{h}_i$ are also Uncorrelated $\forall i$}
\begin{theorem}
The normalized average rate in FFR system is given by
\begin{equation}
R_f = \int\limits_{0}^{R}\underset{t>0}\int\left(\underset{j \in \psi}\prod\frac{1}{1+\max\{e^t-1,T,S_{th}\} r^\alpha d_j^{-\alpha}}+\frac{1}{3}\underset{i \in \phi}\prod\frac{P[\eta(r)<S_{th}]}{1+\max\{e^t-1,T\}r^\alpha d_i^{-\alpha}}\right)\mathrm{d}tf_R(r)\mathrm{d}r \label{eq:ffr_rat}.
\end{equation}
\end{theorem}
\begin{proof}
Since a cell-centre user is a user with $\eta(r)>S_{th}$, the normalized average rate of cell-centre users in FFR system $R_c(r)$ can be written as
\begin{equation}
R_c(r)=E[\ln(1+\eta(r))|\eta(r)>T,\eta(r)>S_{th}]P[\eta(r)>T|\eta(r)>S_{th}]\label{eq:centre_rate3}.
\end{equation}
Similarly, since a cell-edge user is a user with $\eta(r)<S_{th}$, the normalized average rate of cell-edge users in FFR system $R_e(r)$ can be written as
\begin{equation}
R_e(r)=E[\ln(1+\hat{\eta}(r))|\hat{\eta}(r)>T,\eta(r)<S_{th}]P[\hat{\eta}(r)>T|\eta(r)<S_{th}]\label{eq:edge_rate3}.
\end{equation}
Now, the normalized average rate  in FFR system $R_f(r)$ can be written as
\begin{eqnarray}
R_f(r)=R_c(r)P[\eta(r)>S_{th}]+
\frac{1}{3}R_e(r)P[\eta(r)<S_{th}].
\label{eq:ffr_rate}
\end{eqnarray}
Here the first term  denotes the normalized average rate  contributed by cell-centre users, and the second term  denotes the contribution from cell-edge users. Using the methods outlined in section IV$.A$, $R_c(r)P[\eta(r)>S_{th}]$  can be written as
\begin{eqnarray*}
R_c(r)P[\eta(r)>S_{th}]=\underset{t>0}\int P[\ln(1+\eta(r))>t,\eta(r)>T,\eta(r)>S_{th}]\text{dt}
\end{eqnarray*}
\begin{eqnarray}
 =\underset{t>0}\int P[\eta(r)>\max\{e^t-1,T,S_{th}\}]\text{dt} \label{eq:rate_cen}.
\end{eqnarray}
Using \eqref{eq:fr1} and \eqref{eq:reuse1}, this can be further simplified as
\begin{equation}
R_c(r)P[\eta(r)>S_{th}]=\underset{t>0}\int\underset{j \in \psi}\prod\frac{1}{1+\max\{e^t-1,T,S_{th}\} r^\alpha d_j^{-\alpha}}\text{d}t\label{eq:centre_rate1}.
\end{equation}
Again, similar to section IV$.A$ using the fact that $\ln(1+\hat{\eta}(r))$ is a positive RV and monotonic increasing function of $\hat{\eta}(r)$, one can write $R_e(r)$ as
\begin{eqnarray*}
R_e(r)=\underset{t>0}\int \frac{P[\ln(1+\hat{\eta}(r))>t,\hat{\eta}(r)>T,\eta(r)<S_{th}]}{P[\eta(r)<S_{th}]}\text{dt}.
\end{eqnarray*}
Further simplifying $R_e(r)$, one obtains
\begin{equation}
R_e(r)=\underset{t>0}\int\frac{ P[\hat{\eta}(r)>\max\{e^t-1,T\},\eta(r)<S_{th}]}{{P[\eta(r)<S_{th}]}}\text{dt} \label{eq:edge_rate4}.
\end{equation}
Since $g$ and $\hat{g}$ are i.i.d and also $h_i$ and $\hat{h}_i$ are i.i.d and hence, $R_e(r)$ can be written as
\begin{equation}
R_e(r)= \underset{t>0}\int\underset{i \in \phi}\prod\frac{1}{1+\text{max}\{e^t-1,T\}r^\alpha d_i^{-\alpha}}\mathrm{d}t \label{approx_rate1}.
\end{equation}
Recalling the expression for $R_3$ given in \eqref{fr3}, one can see that the normalized average rate of cell-edge users in FFR system is  equal to the normalized rate of FR$3$ system. 
Finally putting back   \eqref{eq:centre_rate1} and \eqref{approx_rate1} into \eqref{eq:ffr_rate} and after averaging over the spatial dimension, 
the normalized average rate in FFR system is given by
\begin{equation}
R_f = \int\limits_{0}^{R}\underset{t>0}\int\left(\underset{j \in \psi}\prod\frac{1}{1+\max\{e^t-1,T,S_{th}\} r^\alpha d_j^{-\alpha}}+\frac{1}{3}\underset{i \in \phi}\prod\frac{P[\eta(r)<S_{th}]}{1+\max\{e^t-1,T\}r^\alpha d_i^{-\alpha}}\right)\mathrm{d}tf_R(r)\mathrm{d}r \label{eq:ffr_rat1}.
\end{equation}
\end{proof}

\subsection{Optimum value of SIR Threshold $S_{opt,R}$ when $g$ and $\hat{g}$ are uncorrelated and $h_i$ and $\hat{h_i}$ are also uncorrelated $\forall i$}
The optimum value of $S_{th}$ (denoted by $S_{opt,R}$) for which the normalized average rate in FFR system is maximized is derived and it is shown to be a function of $T$ and path loss exponent.
\begin{theorem}
The value of $S_{th}$ which maximize the normalized average rate in the FFR system is $S_{opt,R}=\max(T,T')$, where $T'$ can be obtained as the solution of following equation
\begin{align}
\int\limits_{0}^{R}\left(\frac{\left(K(r)-\ln{(1+T')}\right)\sum\limits_{i\in \psi}r^\alpha d_i^{-\alpha}\left ( \underset{j \in \psi\setminus i}\prod{(1+T' r^\alpha d_j^{-\alpha}})\right)}{\left(\underset{j \in \psi}\prod{(1+T' r^\alpha d_j^{-\alpha})}\right)^2}\right)f_R(r)\text{d}r=0\label{delta3},
\end{align}
here, $K(r)$ is defined in \eqref{approx}.
\end{theorem}
\begin{proof}
To obtain the $S_{opt,R}$, we consider the three possibilities:  $(i).$ $S_{th}<T$, $(ii).$ $S_{th}=T$, and $(iii).$ $S_{th}>T$.

{\em Case $(i)$ $S_{th}<T$}:$ \text{ Let }S_{th}=T-\Delta, $ where $\Delta>0$, then $R_f$  can be given by
\begin{equation}
 R_f = \int\limits_{0}^{R}\underset{t>0}\int\left(\underset{j \in \psi}\prod\frac{1}{1+\max\{e^t-1,T\} r^\alpha d_j^{-\alpha}}+\frac{1}{3}\underset{i \in \phi}\prod\frac{P[\eta(r)<T-\Delta]}{1+\max\{e^t-1,T\}r^\alpha d_i^{-\alpha}}\right)\mathrm{d}tf_R(r)\mathrm{d}r\label{eq:compare1}.
\end{equation}
{\em Case $(ii)$ $S_{th}=T$}: $R_f$ is given by
\begin{equation}
 R_f = \int\limits_{0}^{R}\underset{t>0}\int\left(\underset{j \in \psi}\prod\frac{1}{1+\max\{e^t-1,T\} r^\alpha d_j^{-\alpha}}+\frac{1}{3}\underset{i \in \phi}\prod\frac{P[\eta(r)<T]}{1+\max\{e^t-1,T\}r^\alpha d_i^{-\alpha}}\right)\mathrm{d}tf_R(r)\mathrm{d}r.\label{eq:compare2}
\end{equation}
It can be observed that the first term of integrand in both \eqref{eq:compare1} and \eqref{eq:compare2} is the same. However, the second term of integrand in \eqref{eq:compare1} is lower than the second term of integrand in \eqref{eq:compare2} since $P[\eta(r)<T-\Delta]<P[\eta(r)<T]$. Hence, the normalized average rate of a FFR system is higher when  $S_{th}=T$  than when $S_{th}<T$.
Now, let us compare cases $(ii)$ and $(iii)$.

{\em Case $(iii)$  $S_{th}>T$}: $\text{Let }S_{th}=T+\Delta$, then $R_f$ can be given by with $\Delta>0$
\begin{equation}
 R_f = \int\limits_{0}^{R}\underset{t>0}\int\left(\underset{j \in \psi}\prod\frac{1}{1+\max\{e^t-1,T+\Delta\} r^\alpha d_j^{-\alpha}}+\frac{1}{3}\underset{i \in \phi}\prod\frac{P[\eta(r)<T+\Delta]}{1+\max\{e^t-1,T\}r^\alpha d_i^{-\alpha}}\right)\mathrm{d}tf_R(r)\mathrm{d}r\label{eq:compare3}.
\end{equation}
It is apparent  that the first term in \eqref{eq:compare2} is higher than the first term in \eqref{eq:compare3} while the second term in \eqref{eq:compare2} is lower than the corresponding term in \eqref{eq:compare3}. Hence, we combine \eqref{eq:compare2} and \eqref{eq:compare3} and find the $S_{opt,R}$ when $S_{th}\geq T$. Basically, we need to maximize the $R_f( S_{th}\geq T)$ given in \eqref{eq:compare4} with respect to $S_{th}$ for a given inequality $S_{th}\geq T$. Since $R_f( S_{th}\geq T)$ is differentiable with respect to $S_{th}$, we use Karush-Kuhn-Tucker (KKT) conditions to obtain the $S_{opt,R}$ \cite{boyd2004convex}.   Combining \eqref{eq:compare2} and \eqref{eq:compare3}, the normalized average rate expression  can be written as 
\begin{equation}
\textstyle
R_f( S_{th}\geq T) = \int\limits_{0}^{R}\underset{t>0}\int\left(\underset{j \in \psi}\prod\frac{1}{1+\max\{e^t-1,S_{th}\} r^\alpha d_j^{-\alpha}}+\frac{1}{3}\underset{i \in \phi}\prod\frac{P[\eta(r)<S_{th}]}{1+\max\{e^t-1,T\}r^\alpha d_i^{-\alpha}}\right)\mathrm{d}tf_R(r)\mathrm{d}r\label{eq:compare4}.
\end{equation}
To maximize the $R_f( S_{th}\geq T)$, we need to maximize the cost function $J$ and it is given by 
\begin{equation}
J=R_f( S_{th}\geq T)+\lambda(S_{th}-T),
\end{equation}
here $\lambda$ is Lagrange multiplier associated  with the inequality constraint $S_{th}\geq T$. The optimal value of $S_{th}$, i.e, $S_{opt,R}$ must satisfy the  KKT necessary  conditions and they are given by
\begin{align}
&\frac{dJ}{d S_{th}}=0,\label{kkt1}\\
&\lambda\geq 0,\\
& \lambda(S_{th}-T)=0.\label{kkt3}
\end{align}
To solve \eqref{kkt1}, we split the first part of integrand of $R_f( S_{th}\geq T) $ as follows:
\begin{equation}
\textstyle
\underset{t>0}\int\underset{j \in \psi}\prod\frac{1}{1+\max\{e^t-1,S_{th}\} r^\alpha d_j^{-\alpha}}\text{d}t=\overset{\ln(1+S_{th})}{\underset{t>0}\int}\underset{j \in \psi}\prod\frac{1}{1+S_{th} r^\alpha d_j^{-\alpha}}\text{d}t +\overset{\infty}{\underset{\ln(1+S_{th})}\int}\underset{j \in \psi}\prod\frac{1}{1+(e^t-1) r^\alpha d_j^{-\alpha}}\text{d}t
\end{equation}
Also, putting  $P[\eta(r)<S_{th}]=\left(1-\underset{j \in \psi}\prod\frac{1}{1+S_{th} r^\alpha d_j^{-\alpha}}\right)$,  $R_f( S_{th}\geq T)$ can be rewritten in the following form
\begin{align}
  R_f(S_{th}\geq T)= &\int\limits_{0}^{R}\Bigg(\underset{j \in \psi}\prod\frac{\ln(1+S_{th})}{1+S_{th} r^\alpha d_j^{-\alpha}} +\overset{\infty}{\underset{\ln(1+S_{th})}\int}\underset{j \in \psi}\prod\frac{1}{1+(e^t-1) r^\alpha d_j^{-\alpha}}\mathrm{d}t\nonumber\\
&+\left(1-\underset{j \in \psi}\prod\frac{1}{1+S_{th} r^\alpha d_j^{-\alpha}}\right)\underbrace{\frac{1}{3}\underset{t>0}\int \underset{i \in \phi}\prod\frac{1}{1+\max\{e^t-1,T\}r^\alpha d_i^{-\alpha}}\mathrm{d}t}_{K(T,r)}\Bigg)f_R(r)\mathrm{d}r.
\end{align}
Using Leibniz's rule\footnote{Leibniz's rule states that if $f(x,\theta)$ is a function such that $\frac{d}{d \theta }f(x,\theta)$ exist, and is continuous, then
$\frac{d}{d\theta} \left (\int_{a(\theta)}^{b(\theta)} f(x,\theta)\,dx \right )= \int_{a(\theta)}^{b(\theta)} \frac{d}{d \theta }(f(x,\theta))\,dx + f(b(\theta),\theta)\frac{d}{d \theta }b(\theta)-f(a(\theta),\theta)\frac{d}{d \theta }a(\theta)$.} while differentiating $R_f(S_{th}\geq T)$ with respect to $S_{th}$, $\frac{dJ}{d S_{th}}$ can be written as 
\begin{equation*}
\frac{dJ}{d S_{th}}=\lambda+ \int\limits_{0}^{R}\Bigg(\frac{\frac{\underset{j \in \psi}\prod{(1+S_{th} r^\alpha d_j^{-\alpha}})}{1+S_{th}}-\ln{(1+S_{th})}\frac{d}{dS_{th}} \left ( \underset{j \in \psi}\prod{(1+S_{th} r^\alpha d_j^{-\alpha}})\right)}{\left(\underset{j \in \psi}\prod{(1+S_{th} r^\alpha d_j^{-\alpha})}\right)^2}
\end{equation*}
\begin{equation*}
-\underset{j \in \psi}\prod\frac{1}{1+S_{th} r^\alpha d_j^{-\alpha}}\left(\frac{1}{1+S_{th}}\right)+\frac{K(T,r)\frac{d}{dS_{th}} \left ( \underset{j \in \psi}\prod{(1+S_{th} r^\alpha d_j^{-\alpha}})\right)}{\left(\underset{j \in \psi}\prod{(1+S_{th} r^\alpha d_j^{-\alpha})}\right)^2}\Bigg)f_R(r)\mathrm{d}r.
\end{equation*}
Simplifying $\frac{dJ}{d S_{th}}$ and equating it to zero in  accordance with  KKT conditions in \eqref{kkt1}  one obtains
\begin{align*}
\textstyle
\frac{dJ}{d S_{th}}=\lambda+\int\limits_{0}^{R}\left(\frac{K(T,r)\frac{d}{dS_{th}} \left ( \underset{j \in \psi}\prod{(1+S_{th} r^\alpha d_j^{-\alpha}})\right)}{\left(\underset{j \in \psi}\prod{(1+S_{th} r^\alpha d_j^{-\alpha})}\right)^2}-\frac{\ln{(1+S_{th})}\frac{d}{dS_{th}} \left ( \underset{j \in \psi}\prod{(1+S_{th} r^\alpha d_j^{-\alpha}})\right)}{\left(\underset{j \in \psi}\prod{(1+S_{th} r^\alpha d_j^{-\alpha})}\right)^2}\right)f_R(r)\text{d}r=0.
\end{align*}
Further simplifying,  one obtains
\begin{align}
\frac{dJ}{dS_{th}}=\lambda+\int\limits_{0}^{R}\left(\frac{\left(K(T,r)-\ln{(1+S_{th})}\right)\sum\limits_{i\in \psi}r^\alpha d_i^{-\alpha}\left ( \underset{j \in \psi\setminus i}\prod{(1+S_{th} r^\alpha d_j^{-\alpha}})\right)}{\left(\underset{j \in \psi}\prod{(1+S_{th} r^\alpha d_j^{-\alpha})}\right)^2}\right)f_R(r)\text{d}r=0\label{delta}
\end{align}
Now we consider two cases: $(i)$ when $S_{th}>T$ and $(ii)$ when $S_{th}=T$

{\em Case $(i)$ $S_{th}>T$}:  Note that $ \lambda(S_{th}-T)=0$ (from \eqref{kkt3}) and since $S_{th}>T$, this implies  $\lambda=0$. Thus when $S_{th}>T$, Eq. \eqref{delta} can be written as
\begin{align}
\int\limits_{0}^{R}\left(\frac{\left(K(T,r)-\ln{(1+S_{th})}\right)\sum\limits_{i\in \psi}r^\alpha d_i^{-\alpha}\left ( \underset{j \in \psi\setminus i}\prod{(1+S_{th} r^\alpha d_i^{-\alpha}})\right)}{\left(\underset{j \in \psi}\prod{(1+S_{th} r^\alpha d_j^{-\alpha})}\right)^2}\right)f_R(r)\text{d}r=0\label{delta1}
\end{align}
Solving $\eqref{delta1}$ analytically in its current form to obtain $S_{opt,R}$ is a difficult problem since $d_i$s are also the function of $r$.  However, we can solve \eqref{delta1} by exploiting the following observation that $K(T,r)$ is nearly independent of $T$ for sufficiently small value of $T$ and it can be approximated by $K(r)$, which is given by
\begin{equation}
K(T,r)=\frac{1}{3}\underset{t>0}\int \underset{i \in \phi}\prod\frac{1}{1+\max\{e^t-1,T\}r^\alpha d_i^{-\alpha}}\mathrm{d}t\approx \frac{1}{3}\underset{t>0}\int \underset{i \in \phi}\prod\frac{1}{1+(e^t-1)r^\alpha d_i^{-\alpha}}\mathrm{d}t=K(r)\label{approx}
\end{equation}
The approximation $K(T,r)\approx  K(r)$ in \eqref{approx} is possible since $r^\alpha d_i^{-\alpha}\ll 1$ for $i\in \phi$, because in a planned network $d_i$ are the distances from the $2$nd tier BS. Note that this approximation is not possible when $\delta<3$ or in an unplanned network. To show the tightness of the approximation we plot the values of $K(T,r)$ and $K(r)$ with respect to $T$ for three different values of $r$ as shown in Fig. \ref{fig:fig7}. It can be observed that the value of $K(T,r)$ is very close to the value of $K(r)$ for small value of $T$. Thus making use of $K(T,r)\approx K(r)$  one can rewrite \eqref{delta1} as 
\begin{figure}[ht]
 \centering
 \includegraphics[scale=0.3]{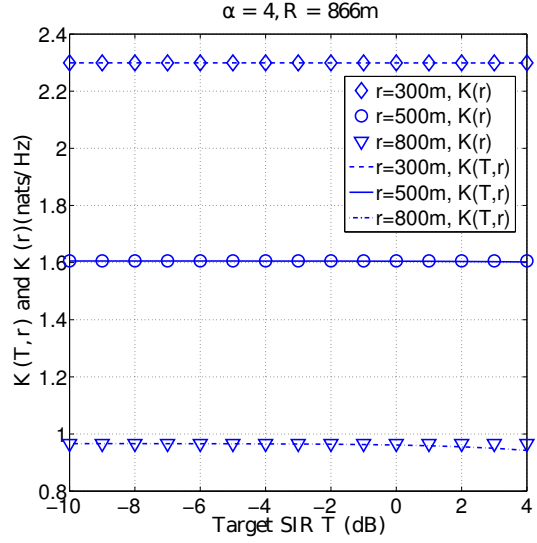}
 \caption{Variation in $K(T,r)$ and $K(r)$ with respect to target SIR $T$.}
 \label{fig:fig7}
 \end{figure}
\begin{align}
\int\limits_{0}^{R}\left(\frac{\left(K(r)-\ln{(1+S_{th})}\right)\sum\limits_{i\in \psi}r^\alpha d_i^{-\alpha}\left ( \underset{j \in \psi\setminus i}\prod{(1+S_{th} r^\alpha d_j^{-\alpha}})\right)}{\left(\underset{j \in \psi}\prod{(1+S_{th} r^\alpha d_j^{-\alpha})}\right)^2}\right)f_R(r)\text{d}r=0\label{delta2}
\end{align}
The solution of above integral equation is fairly simple. Hence, one can obtain the value of $S_{opt,R}$ by solving \eqref{delta2} numerically (using Mathematica). One can observe from \eqref{delta2} that as $K(r)$  will increase, $S_{opt,R}$ will also increase. Intuitively it is true since  $K(r)$ denote the rate achieved by FR$3$ and as rate of FR$3$ will increase percentage of cell-edge user should also increase and hence $S_{opt,R}$.   We have observed that $S_{opt,R}$ as a function of the path loss exponent, and it is constant with respect to $T$ (when $S_{th}>T$). We denote $S_{opt,R}$ by $T'$ for the case when $S_{th}>T$.

{\em Case $(ii)\text{ } S_{th}=T$}: In this case $\lambda>0$ since $ \lambda(S_{th}-T)=0$ and hence $S_{opt,R}=T$

Hence based on the solution obtained for both the cases the optimal value of $S_{th}$ is given by $S_{opt,R}=\max(T,T')$ where $T'$ is a function of the path loss exponent.
\end{proof} 
We also numerically evaluate the expression  in \eqref{eq:ffr_rat} and show that the numerical values match with our theoretical derivation (Note that in \eqref{eq:ffr_rat} we do not approximate $K(T,r)$). Fig. \ref{fig:ffr_rate} plots the normalized average rate of FR$1$, FR$3$, and FFR with respect to $S_{th}$ for three values of $T$ namely, $0$dB, $1$dB, and $2$dB. We see that for $T=1$dB and $2$dB, using $S_{th}=T$ gives the maximum normalized average rate. However, for $T=0$dB, $S_{th}=1$dB gives the maximum normalized average rate. In other words, there exists a  SIR $T'$, such that for $T\geq T'$, $S_{th}=T$ performs better than $S_{th}>T$. Also, for $T<T'$, $S_{th}=T'$ gives the maximum normalized rate. Therefore,  $S_{th}=\max(T, T')$ gives the maximum normalized average rate and it match our theoretical observation.  \begin{figure*}
    
        \subfigure[Maximum rate at $S_{th}=1$dB]
        {
            \includegraphics[width=2in]{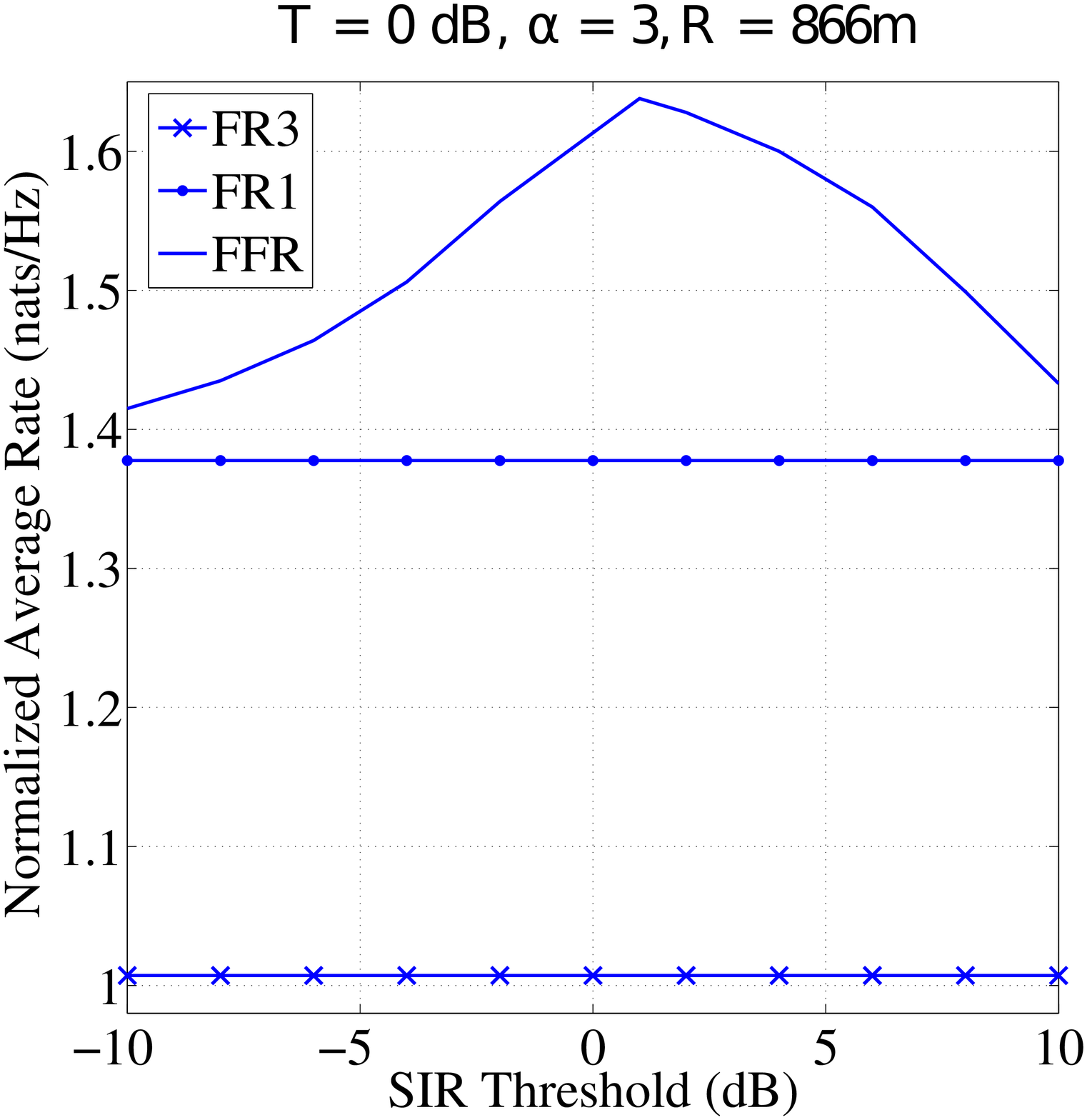}
            \label{fig:first_sub}
        }
        \subfigure[Maximum rate at $S_{th}=1$dB]
        {
            \includegraphics[width=2in]{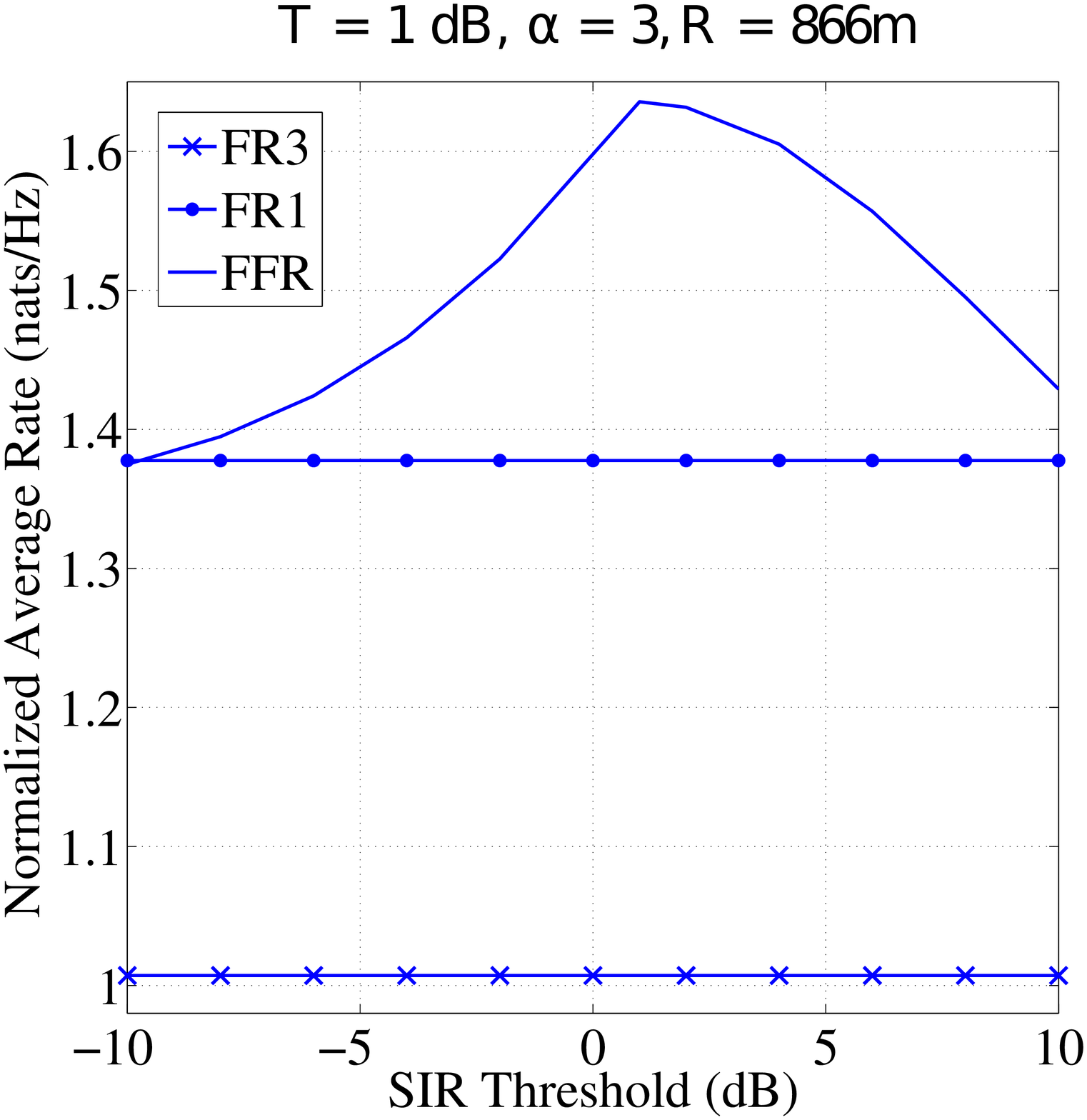}
            \label{fig:second_sub}
        }
        \subfigure[Maximum rate at $S_{th}=2$dB]
        {
            \includegraphics[width=2in]{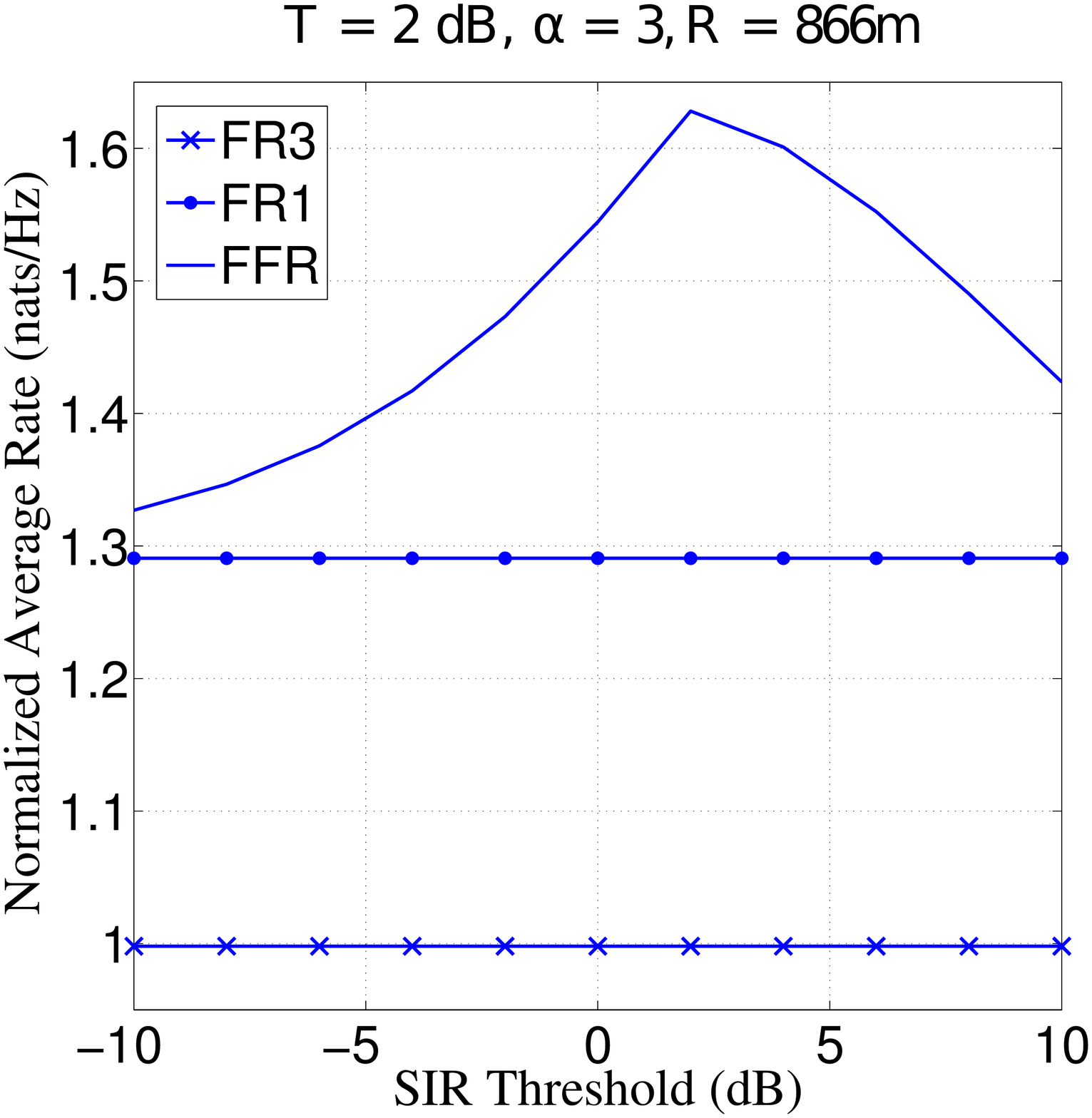}
            \label{fig:third_sub}
        }
        \caption{Normalized spectral efficiency of FR$1$, FR$3$, and FFR with respect to SIR threshold $S_{th}$ for three different values of Target SIR $T$ when fading are independent across the sub-bands.}
  \label{fig:ffr_rate}
    \end{figure*}
 \begin{figure}[ht]
 \centering
 \includegraphics[scale=0.35]{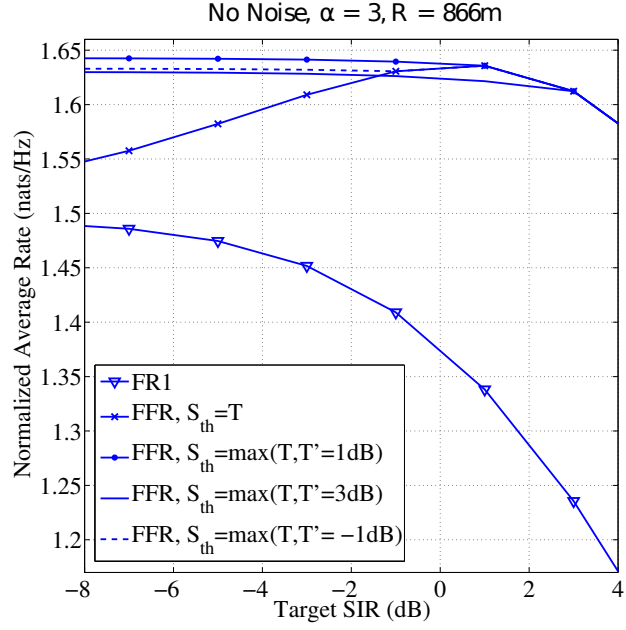}
 \caption{Normalized spectral efficiency of FR$1$ and FFR with respect to target SIR $T$ when fading are independent across the sub-band.}
 \label{fig:fig4}
 \end{figure}

Another numerical result is also   presented in Fig. \ref{fig:fig4}, where the normalized average rate as a function of SIR  target  $T$ is plotted  for three different choices of $T'$ namely $-1$dB, $1$dB and $3$dB, and  $S_{th}$ is chosen as $\max(T,T')$. A fourth choice for $S_{th}$, namely $S_{th}=T$ is also shown in Fig. \ref{fig:fig4} for comparison. From these curves, which are numerically evaluated using \eqref{eq:ffr_rat}, it seems that $T'=1$dB provides the best rate performance for the considered FFR based cellular network. A lower value of $T'$ may push too many users from cell-edge to cell-centre resources, while a higher value of $T'$ may increase the number of cell-edge users. In either case, the  normalized average rate (which taken into account the rates of all the users in the network including those with zero rate) will be lower. 
\subsection{Normalized Average Rate of FFR System When the Sub-bands are Completely Correlated}
In this subsection first we derive the normalized average rate of FFR system. The normalized average rate  in FFR system $R_f(r)$ given in \eqref{eq:ffr_rate} can be rewritten as
\begin{eqnarray}
R_f(r)=R_c(r)P[\eta(r)>S_{th}]+
\frac{1}{3}R_e(r)P[\eta(r)<S_{th}].
\label{eq:ffr_rate1}
\end{eqnarray}
Note that the first term $R_c(r)P[\eta(r)>S_{th}]$ denotes the normalized average rate  contributed by cell-centre users  and it is same for both the cases (fading are correlated or independent across the sub-bands). 
Now, using the expression of $R_e(r)$ given in \eqref{eq:edge_rate4}, $R_e(r)P[\eta(r)<S_{th}]$  can be written as
\begin{equation}
R_e(r)P[\eta(r)<S_{th}]=\underset{t>0}\int P[\hat{\eta}(r)>\max\{e^t-1,T\},\eta(r)<S_{th}]\text{dt} \label{eq:edge_rate}.
\end{equation}
Using the transformation given in \eqref{eq:tran}, $R_e(r)P[\eta(r)<S_{th}]$ can be  simplified as
\begin{equation}
R_e(r)P[\eta(r)<S_{th}]=\underset{t>0}\int P[\hat{\eta}(r)>\max\{e^t-1,T\}]-P[\hat{\eta}(r)>\max\{e^t-1,T,\hat{S}_{th}\}]\text{dt} \label{eq:edge_rate1}. 
\end{equation}
Using the result given in \eqref{fr3}, $R_e(r)P[\eta(r)<S_{th}]$ can be further  simplified as
\begin{equation}
R_e(r)P[\eta(r)<S_{th}]=\underset{t>0}\int \underset{i \in \phi}\prod\frac{1}{1+\text{max}\{e^t-1,T\}r^\alpha d_i^{-\alpha}}-\underset{i \in \phi}\prod\frac{1}{1+\text{max}\{e^t-1,T,\hat{S}_{th}\}r^\alpha d_i^{-\alpha}}\text{dt} \label{approx_rate2}. 
\end{equation}
Finally putting back   \eqref{eq:centre_rate1} and \eqref{approx_rate2} into \eqref{eq:ffr_rate1} and after averaging over the spatial dimension, 
the normalized average rate in FFR system is given by
\begin{equation}
\scriptstyle
R_f = \int\limits_{0}^{R}\underset{t>0}\int\underset{j \in \psi}\prod\frac{1}{1+\max\{e^t-1,T,S_{th}\} r^\alpha d_j^{-\alpha}}+\frac{1}{3} \left(\underset{i \in \phi}\prod\frac{1}{1+\text{max}\{e^t-1,T\}r^\alpha d_i^{-\alpha}}-\underset{i \in \phi}\prod\frac{1}{1+\text{max}\{e^t-1,T,\hat{S}_{th}\}r^\alpha d_i^{-\alpha}}\right)\mathrm{d}tf_R(r)\mathrm{d}r \label{eq:ffr_rat2}.
\end{equation}
 \begin{figure*}
    
        \subfigure[Maximum rate at $S_{th}=1$dB]
        {
            \includegraphics[width=2in]{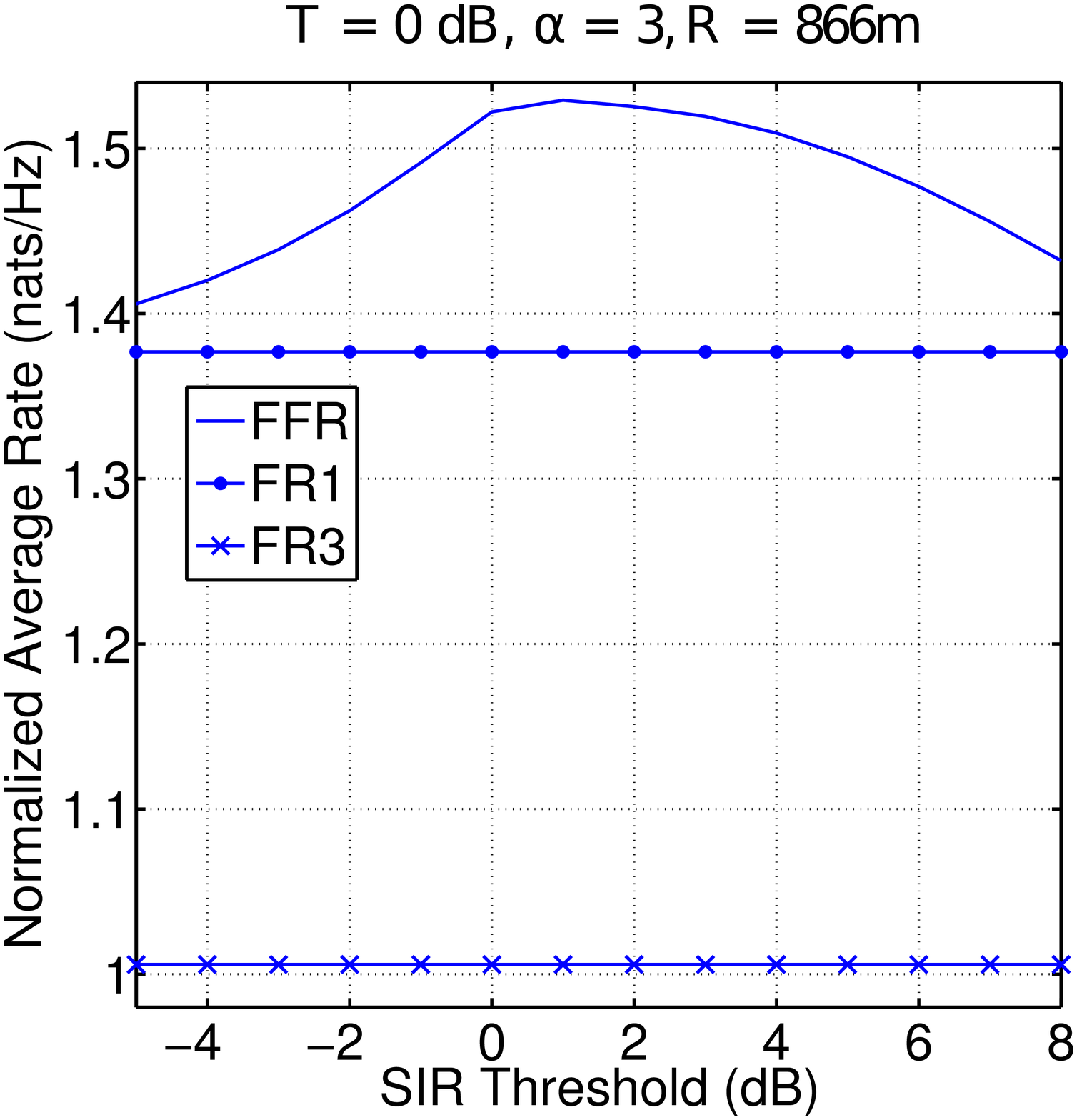}
            \label{fig:first_sub1}
        }
        \subfigure[Maximum rate at $S_{th}=1$dB]
        {
            \includegraphics[width=2in]{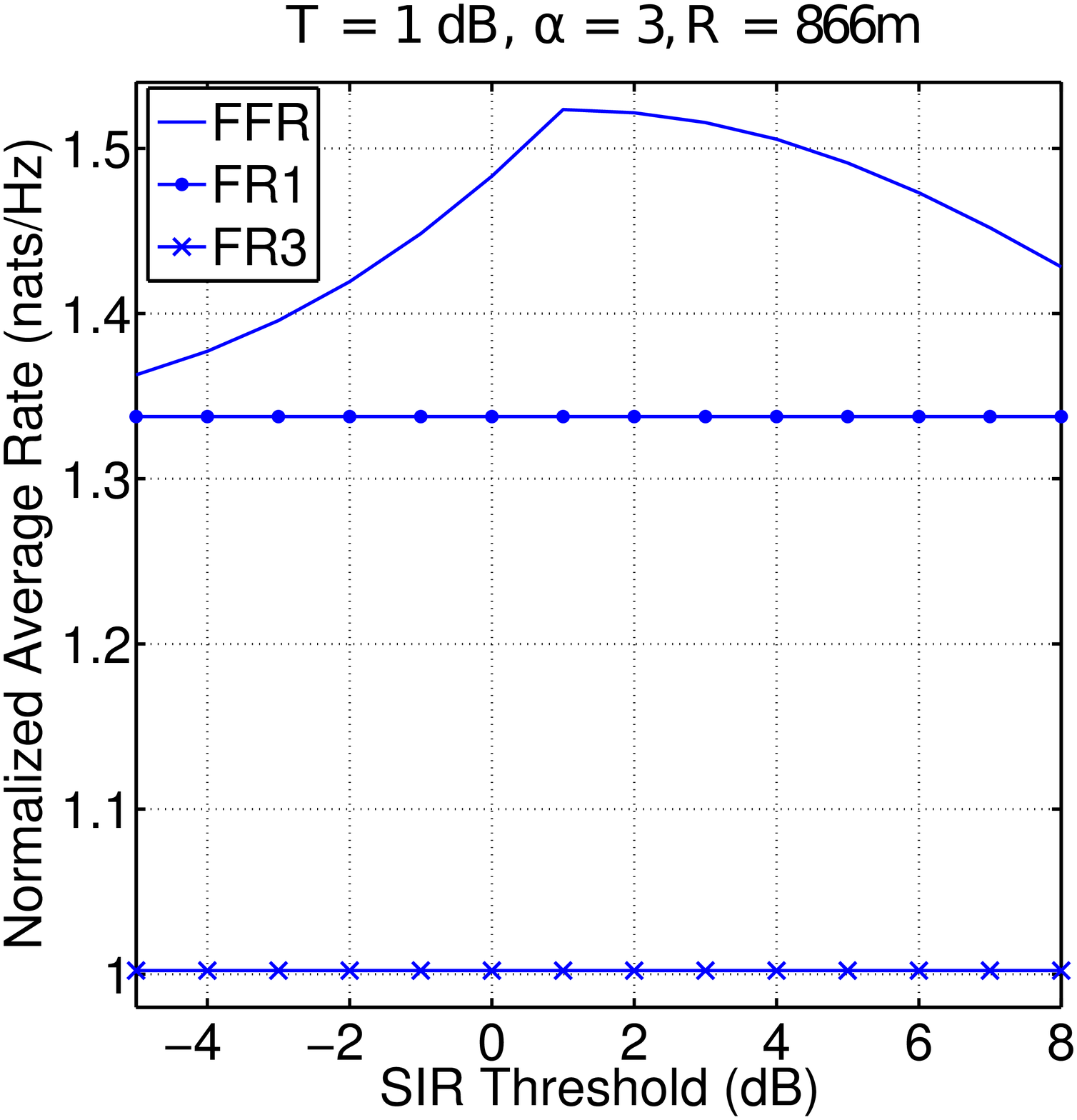}
            \label{fig:second_sub1}
        }
        \subfigure[Maximum rate at $S_{th}=2$dB]
        {
            \includegraphics[width=2in]{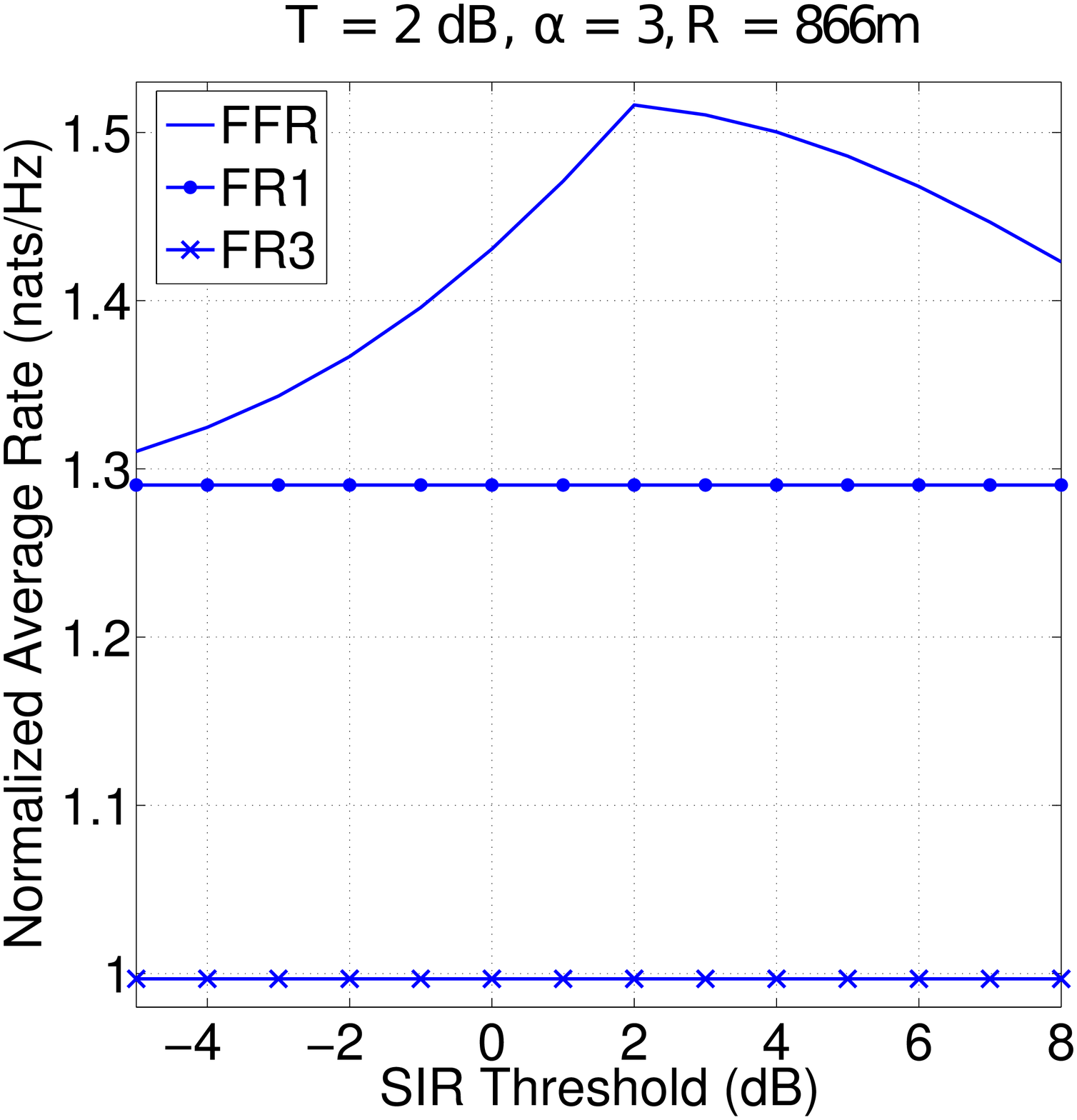}
            \label{fig:third_sub1}
        }
        \caption{Normalized spectral efficiency of FR$1$, FR$3$, and FFR with respect to SIR threshold $S_{th}$ for three different values of Target SIR $T$ when fading are fully correlated across the sub-band.}
  \label{fig:ffr_rate1}
    \end{figure*}
We numerically evaluate the expression  in \eqref{eq:ffr_rat2}.  Fig. \ref{fig:ffr_rate} plots the normalized average rate of FR$1$, FR$3$, and FFR with respect to $S_{th}$ for three values of $T$ namely, $0$dB, $1$dB, and $2$dB. We see that for $T=1$dB and $2$dB, using $S_{th}=T$ gives the maximum normalized average rate. However, for $T=0$dB, $S_{th}=1$dB gives the maximum normalized average rate. In other words, there exists a  SIR $T''$, such that for $T\geq T''$, $S_{th}=T$ performs better than $S_{th}>T$. Also, for $T<T''$, $S_{th}=T''$ gives the maximum normalized rate. Therefore,  $S_{th}=\max(T, T'')$ gives the maximum normalized average rate when fading channel gains are fully correlated across the sub-bands.

Table I shows the variation in $T'$, $T''$, percentage of users classified as cell-centre users and percentage gain in normalized average rate of FFR when compared with normalized average rate of FR$1$ as a function of the path loss exponent. Here to obtain percentage gain in normalized average rate we assume $T=T'$ for case $(i)$ and $T=T''$ for case $(ii)$,  and percentage of cell-centre users is nothing but $P[\text{SIR}>S_{opt,R}]$. First, we note that  $T'\approx T''$, however the normalized average rate achieved at $S_{opt,R}$ using independent fading is higher than the case when fading are fully correlated.  Secondly, it can be observed that as $\alpha$ increases both $T'$ and $T''$  and also percentage of cell-centre users increases.  However, percentage  gain in normalized average rate decreases as path loss exponent increases and percentage gain in normalized rate is higher when fading are independent than the case when fading are fully correlated.
\begin{table}[ht]
        \centering
\caption{Variation in $T'$, $T''$, percentage of cell-centre users and percentage gain in normalized average rate with respect to path loss exponent for both the cases: case $(i)$ when fading are independent. and case $(ii)$ when fading are fully correlated.} 
\renewcommand{\tabcolsep}{0.35cm}
\renewcommand{\arraystretch}{1.5}
\begin{tabular}{|c| c| c| c|c |c| }
\hline 
$\alpha$  & {\bf $T'$} for case $(i)$&{\bf $T''$}for case $(ii)$ &{ $\%$ Cell-centre users} & $\%$ gain in case $(i)$ & $\%$ gain in case $(ii)$  \tabularnewline
\hline
\hline 
2 & $-2.3$dB& $-2.5$dB &50 $\%$ &31.6 $\%$ &16.65  $\%$\tabularnewline
\hline 
2.5& $-0.5$dB & $-0.6$dB & 53$\%$ &26.2  $\%$ &15.2  $\%$\tabularnewline
\hline 
3 &  $1$dB &  $1$dB& 56 $\%$ &22.2  $\%$ &13.9  $\%$\tabularnewline
\hline 
3.5&  $2.3$dB &  $2.3$dB& 59 $\%$ &19.4  $\%$ & 13  $\%$\tabularnewline
\hline 
4  &  $3.5$dB &  $3.5$dB & 62 $\%$ &17.5 $\%$ & 12.4  $\%$\tabularnewline
\hline 
\end{tabular}
\end{table} 

Finally, we have two different expressions for optimal SINR threshold for both the cases, one corresponding to coverage probability ($S_{th}= T $) and other corresponding to normalized average rate ($S_{th}=\max(T,T')$). To maximize  both coverage probability as well as  normalized average rate simultaneously, the system designer may choose  either one of these two expressions.
However, if we choose $S_{th}=T$, then the  normalized average rate decreases significantly as shown in Fig. \ref{fig:fig4}, on the other hand when we choose $S_{th}=\max(T,T')$, normalized average rate is maximized and the loss in coverage probability over choosing $S_{th}=T$ is negligible. Therefore,  one could choose $S_{th}=\max(T,T')$ to maximize  both coverage probability as well as the normalized average rate.
\section{Conclusion}
This work has derived expressions for coverage probability and normalized average rate for OFDMA system utilizing planned FFR deployment. The impact of frequency correlation between the sub-bands allocated to FR$1$ and FR$3$ regions on the average rate and the coverage probability is analysed in detail, since any practical OFDMA system will typically see some correlation. We analytically obtained the  optimal SINR threshold which maximizes the coverage probability, and also determined the  optimal SINR threshold which maximizes the normalized average rate for the following cases: $(i)$ sub-bands are uncorrelated and $(ii)$ sub-bands are completely correlated. Further, it is shown that for the optimal choice of SINR threshold, coverage probability for the FFR is higher than the FR$3$ coverage probability, and the normalized average rate for FFR is also significantly higher  than the rates achieved by either FR$1$ or FR$3$.
\bibliographystyle{IEEEtran}
\bibliography{bibfile}

\end{document}